\documentclass[11pt]{article}
\usepackage[top=1 in, bottom=1 in, left=1 in, right=1 in, letterpaper]{geometry}
\usepackage{algorithm}
\usepackage[noend]{algorithmic}
\usepackage{amssymb,amsmath}
\usepackage{multirow}
\usepackage{tikz}
\usepackage[affil-it]{authblk}
\usetikzlibrary{arrows,shapes}

\newtheorem{theorem}{Theorem}

\newtheorem{proposition}[theorem]{Proposition}
\newtheorem{corollary}{Corollary}
\newenvironment{proof}{{\bf Proof.}}{\hfill\rule{2mm}{2mm}\\}
\newcommand{\algo}[1]{\textsc{Algorithm~#1}}

\title{Improved Approximation Algorithms\\ for the Non-preemptive Speed-scaling Problem\thanks{Research supported by the
French Agency for Research under the DEFIS program
TODO, ANR-09-EMER-010, by GDR-RO of CNRS, and by THALIS-ALGONOW.}}
\author[1]{Evripidis Bampis}
\author[1]{Giorgio Lucarelli}
\author[1,2]{Ioannis Nemparis}
\affil[1]{LIP6, Universit\'{e} Pierre et Marie Curie, Paris, France\\
\emph{\texttt{\{Evripidis.Bampis,Giorgio.Lucarelli\}@lip6.fr}}}
\affil[2]{Dept. of Informatics and Telecommunications, NKUA, Athens, Greece\\
\emph{\texttt{sdi0700181@di.uoa.gr}}}
\date{}

\begin{document}

\maketitle

\begin{abstract}
We are given a set of jobs, each one specified by its release date, its deadline and its processing volume (work),
and a single (or a set of) speed-scalable processor(s).
We adopt the standard model in speed-scaling in which if a processor runs at speed $s$ then the energy consumption is $s^{\alpha}$ per time unit,
where $\alpha>1$.
Our goal is to find a schedule respecting the release dates and the deadlines of the jobs so that the total energy consumption is minimized.
While most previous works have studied the preemptive case of the problem, where a job may be interrupted and resumed later,
we focus on the non-preemptive case where once a job starts its execution, it has to continue until its completion without any interruption.
%Our main result is a constant factor approximation algorithm for the single-processor speed-scaling problem
%that improves the approximation ratio presented recently in [Antoniadis and Huang, SWAT 2012], for any $\alpha > 1.28$.
We propose improved approximation algorithms for the multiprocessor non-preemptive speed-scaling problem for particular families of instances,
namely instances where all jobs have a common release date (or a common deadline),
instances where all jobs are active at some time, and agreeable instances.
\end{abstract}

\section{Introduction}

One of the main mechanisms used for minimizing the energy consumption in computing systems and portable devices is the so called speed-scaling mechanism.
In this setting, the speed of a processor may change dynamically.
If the speed of the processor is $s(t)$ at time $t$ then its power is $s(t)^{\alpha}$, for some constant $\alpha>1$,
and the energy consumption is the power integrated over time.

There is a series of papers in the speed-scaling literature.
In their seminal paper, Yao et al.~\cite{Yao95} proposed a polynomial-time algorithm
for the energy minimization problem of scheduling a set of $n$ jobs $\mathcal{J}$,
where each job $J_j \in \mathcal{J}$ is characterized by its processing volume (work) $w_j$,
its release date $r_j$ and its deadline $d_j$, on a single speed-scalable processor, assuming
 that the preemption of the jobs, i.e. the possibility to interrupt the execution of a job and resume it later, is allowed.
Most of the subsequent works on energy minimization in the speed-scaling setting allow jobs to preempt.
However in practice, preemption causes an important overhead and it is sometimes even impossible, i.e., in the case where external resources are used.
Hence, it is natural to disallow it.
Only recently, Antoniadis and Huang, in \cite{Antoniadis12}, studied the non-preemptive version of the energy minimization problem
introduced by Yao et al. \cite{Yao95}, and they proved that the problem becomes strongly $\mathcal{NP}$-hard.
They have also proposed a constant factor approximation algorithm.

The same paper also studied the problem for a particular family of instances, namely \emph{laminar} instances.
Here for any two jobs $J_j$ and $J_{j'}$ with $r_j \leq r_{j'}$ it holds that either $d_j \geq d_{j'}$ or $d_j \leq r_{j'}$.
In fact, such instances typically arise when recursive calls in a program create new jobs.
Another interesting family of instances studied in the literature  is the family of \emph{agreeable} instances.
In an agreeable instance, for any two jobs $J_j$ and $J_{j'}$ with $r_j \leq r_{j'}$ it holds that $d_j \leq d_{j'}$,
i.e. latter released jobs have latter deadlines.
Such instances may arise in situations where the goal is to maintain a fair service guarantee for the waiting time of jobs.
Two more special families of instances that we use in our study below are the \emph{clique} instances and the \emph{pure-laminar} instances.
In a clique instance, for any two jobs $J_j$ and $J_{j'}$ with $r_j \leq r_{j'}$ it holds that $d_j \geq r_{j'}$.
In other words, in a clique instance there is a time, $T$, where all jobs are active;
for example $T=\min\{d_j, J_j \in \mathcal{J}\}$ or $T=\max\{r_j, J_j \in \mathcal{J}\}$.
In a pure-laminar instance, for any two jobs $J_j$ and $J_{j'}$ with $r_j \leq r_{j'}$ it holds that $d_j \geq d_{j'}$.
Note that the family of pure-laminar instances is a special case of both laminar and clique instances.
Finally, two other interesting special cases of all the above families, studied by several works in scheduling,
are those where all the jobs have either a common release date or a common deadline.

\subsection{Related work}

As mentioned above, for the preemptive single-processor case, Yao et al. \cite{Yao95} proposed an optimal algorithm
for finding a feasible schedule with minimum energy consumption.
Using an extension of the classical three-field notation, this problem can be denoted as $S1|r_j,d_j,pmtn|E$.
The multiprocessor case, $S|r_j,d_j,pmtn|E$, where there are $m$ available processors has
been solved optimally in polynomial time when preemption and migration of jobs are allowed \cite{Albers11,Angel12,Bingham08}.
The migration assumption means that a job may be interrupted and resumed on the same processor or on another processor.
However, the parallel execution of parts of the same job is not allowed.

Albers et al. \cite{Albers07} considered the multiprocessor problem $S|r_j,d_j,pmtn,no\text{-}mig|E$,
where the preemption of the jobs is allowed, but not their migration.
They first studied the problem where each job has unit work.
They proved that the problem is polynomial time solvable for instances with agreeable deadlines.
%, i.e. for $S|\text{agreeable},w_j=1,pmtn,no\text{-}mig|E$.
For general instances with unit-work jobs, they proved that the problem becomes strongly $\mathcal{NP}$-hard and
they proposed an $(\alpha^{\alpha}2^{4\alpha})$-approximation algorithm.
For the case where the jobs have arbitrary works, the problem was proved to be $\mathcal{NP}$-hard
even for instances with common release dates and common deadlines.
%, i.e. for $S|r_j=0,d_j=d,pmtn,no\text{-}mig|E$.
Albers et al. proposed a $2(2-\frac{1}{m})^{\alpha}$-approximation algorithm for instances with common release dates,
%$S|r_j=0,d_j,pmtn,no\text{-}mig|E$,
or common deadlines,
%$S|r_j,d_j=d,pmtn,no\text{-}mig|E$,
and an $(\alpha^{\alpha}2^{4\alpha})$-approximation algorithm for instances with agreeable deadlines.
%, $S|agreeable,pmtn,no\text{-}mig|E$.
Greiner et al. \cite{Greiner09} gave a generic reduction transforming an optimal schedule for the multiprocessor problem with migration,
$S|r_j,d_j,pmtn|E$, to a $B_{\alpha}$-approximate solution for the multiprocessor problem with preemptions but without migration,
$S|r_j,d_j,pmtn,no\text{-}mig|E$, where $B_{\alpha}$ is the $\alpha$-th Bell number.
%As $S|r_j,d_j,pmtn|E$ can be solved in polynomial time for arbitrary release dates and deadlines,
%it follows that there is a $B_{\alpha}$-approximation algorithm for $S|r_j,d_j,pmtn,no\text{-}mig|E$.
This result holds only when $m \geq \alpha$.

It has to be noticed here that for the family of agreeable instances,
and hence for their special families of instances (instances with common release dates and/or common deadlines),
the assumption of preemption and no migration is equivalent to the non-preemptive assumption that we consider throughout this paper.
More specifically, for agreeable instances, any preemptive schedule can be transformed into a non-preemptive one of the same energy consumption,
where the execution of each job $J_j \in \mathcal{J}$ starts after the completion of any other job which is released before $J_j$.
The correctness of this transformation can be easily proved by induction to the order where the jobs are released.
Hence, the results of \cite{Albers07} and \cite{Greiner09} for agreeable deadlines hold for the non-preemptive case as well.

Finally, the most closely related work  is the work of Antoniadis and Huang \cite{Antoniadis12}
who considered the energy minimization single-processor non-preemptive speed-scaling problem.
They first proved that the problem is $\mathcal{NP}$-hard even for pure-laminar instances.
%The problem is strongly $\mathcal{NP}$-hard even for a particular family of instances, the pure laminar instances.
%An instance of jobs is called pure laminar if for any two jobs $J_j$ and $J_{j'}$ with $r_j \leq r_{j'}$ it holds that $d_j \geq d_{j'}$.
They also presented a $2^{4\alpha-3}$-approximation algorithm for laminar instances and
a $2^{5\alpha-4}$-approximation algorithm for general instances.
Notice that the polynomial-time algorithm for finding an optimal preemptive schedule presented in \cite{Yao95}
returns a non-preemptive schedule when the input instance is agreeable.
%For the multiprocessor case, the problem is known to be $\mathcal{NP}$-hard even for the very special case where all jobs have common release dates and
%common deadlines \cite{Albers07}.

In Table~\ref{tbl:known}, we summarize the most related results of the literature.
Several other results concerning scheduling problems in the speed-scaling setting have been presented, involving
the optimization of some QoS criterion under a budget of energy, or the optimization of a linear combination of the energy consumption and some QoS criterion
(see for example \cite{Bampis12,Bunde06,Pruhs08}).
The interested reader can find more details in the recent survey~\cite{Albers10}.
\begin{table}[htb]
\begin{center}
\begin{footnotesize}
\begin{tabular}{|l|l|l|l|}
\hline
\multirow{2}{*}{Problem} & \multirow{2}{*}{Complexity} & \multicolumn{2}{c|}{Approximation ratio} \\
\cline{3-4}
 & & $m < \alpha$ & $m \geq \alpha$ \\
\hline\hline
$S1|r_j,d_j,pmtn|E$ & Polynomial \cite{Yao95} & \multicolumn{2}{c|}{--} \\
$S|r_j=0,d_j=d,pmtn|E$ & Polynomial \cite{Chen04} & \multicolumn{2}{c|}{--} \\
$S|r_j,d_j,pmtn|E$ & Polynomial \cite{Albers11,Angel12,Bingham08} & \multicolumn{2}{c|}{--} \\
\hline
$S|\mbox{agreeable},w_j=1,pmtn,no\text{-}mig|E$ $^{(*)}$ & Polynomial \cite{Albers07} & \multicolumn{2}{c|}{--} \\
$S|r_j,d_j,w_j=1,pmtn,no\text{-}mig|E$ & $\mathcal{NP}$-hard $(m\geq2)$ \cite{Albers07} & $\alpha^{\alpha} 2^{4\alpha}$ \cite{Albers07} & $B_{\alpha}$ \cite{Greiner09} \\
$S|r_j=0,d_j=d,pmtn,no\text{-}mig|E$ $^{(*)}$ & $\mathcal{NP}$-hard \cite{Albers07} & \multicolumn{2}{c|}{$PTAS$ \cite{Hochbaum87,Albers07}} \\
$S|r_j=0,d_j,pmtn,no\text{-}mig|E$ $^{(*)}$ & $\mathcal{NP}$-hard & $2(2-\frac{1}{m})^{\alpha}$ \cite{Albers07} & $\min\{2(2-\frac{1}{m})^{\alpha}, B_{\alpha}\}$ \cite{Albers07,Greiner09} \\
$S|r_j,d_j=d,pmtn,no\text{-}mig|E$ $^{(*)}$ & $\mathcal{NP}$-hard & $2(2-\frac{1}{m})^{\alpha}$ \cite{Albers07} & $\min\{2(2-\frac{1}{m})^{\alpha}, B_{\alpha}\}$ \cite{Albers07,Greiner09} \\
$S|\mbox{agreeable},pmtn,no\text{-}mig|E$ $^{(*)}$ & $\mathcal{NP}$-hard & $\alpha^{\alpha} 2^{4\alpha}$ \cite{Albers07} & $B_{\alpha}$ \cite{Greiner09} \\
$S|r_j,d_j,pmtn,no\text{-}mig|E$ & $\mathcal{NP}$-hard & -- & $B_{\alpha}$ \cite{Greiner09} \\
\hline
$S1|\mbox{agreeable}|E$ & Polynomial \cite{Yao95,Antoniadis12} & \multicolumn{2}{c|}{--} \\
$S1|\mbox{laminar}|E$ & $\mathcal{NP}$-hard \cite{Antoniadis12} & \multicolumn{2}{c|}{$2^{4\alpha-3}$ \cite{Antoniadis12}} \\
$S1|r_j,d_j|E$ & $\mathcal{NP}$-hard & \multicolumn{2}{c|}{$2^{5\alpha-4}$ \cite{Antoniadis12}} \\
\hline
%$S1|r_j=0,prec,E|C_{\max}$ & Polynomial \cite{Pruhs08} & -- \\
%$S|r_j=0,E|C_{\max}$ & & $PTAS$ \cite{Pruhs08}, $O(1)$ \cite{Pruhs08} \\
%$S|r_j=0,prec|(E,C_{\max})$ & & $(O(\log^2 m),O(\log m))$ \cite{Pruhs08} \\
%$S|r_j=0,prec,E|C_{\max}$ & & $O(\log^{1+2/\alpha} m)$ \cite{Pruhs08} \\
%\hline
%$S1|r_j=0,E|L_{\max}$ & Polynomial \cite{Bampis12} & -- \\
%$S1|r_j=0|E+L_{\max}$ & Polynomial \cite{Bampis12} & -- \\
%$S1|r_j,E|L_{\max}$ & $\mathcal{NP}$-hard \cite{Bampis12} & \\
%$S1|r_j|E+L_{\max}$ & $\mathcal{NP}$-hard \cite{Bampis12} & $2$ \cite{Bampis12} $^{(**)}$ \\
%\hline
\end{tabular}
\end{footnotesize}
\caption{Complexity and approximability results.
$^{(*)}$The problem is equivalent with the corresponding non-preemptive problem.
%$^{(**)}B_{\alpha}$ is the $\alpha$-th Bell number.
}
\label{tbl:known}
\end{center}
\end{table}

\subsection{Our contribution}

In this paper, we explore the approximability for the non-preemptive speed-scaling problem on multiprocessors for special families of instances.
More specifically, in Section~\ref{section:crd-cd} we consider the multiprocessor case where the jobs have either common release dates or common deadlines.
Recall that for these problems algorithms of approximation ratio $2(2-\frac{1}{m})^{\alpha}$ have been presented in \cite{Albers07},
while these results have been improved in \cite{Greiner09} to $B_{\alpha}$, if $\alpha \leq m$ and $\alpha \leq 5$.
We further improve the approximation ratio to $(2-\frac{1}{m})^{\alpha-1}$, for any value of $\alpha$ and $m$.
In Section~\ref{section:agreeable}, we consider the agreeable multiprocessor case and
we present a $(4(2-\frac{1}{m}))^{\alpha-1}$-approximation algorithm.
However, for $\alpha \leq m$, a $B_{\alpha}$-approximation algorithm for the case has been presented in \cite{Greiner09}.
Hence, this result dominates our approximation factor for $\alpha \leq m$.
Even if for practical applications the assumption made in \cite{Greiner09} that $\alpha \leq m$ is justified,
from a theoretical point of view it is worthwhile to investigate the approximability of the problem when $\alpha > m$.
For the latter case, our result improves the ratio of $\alpha^{\alpha}2^{4\alpha}$ given in \cite{Albers07}.

A summary of our results compared with the previously known results is given in Table~\ref{tbl:results}.
\begin{table}[htb]
\begin{center}
\begin{tabular}{|l|l|l|l|}
\hline
\multirow{2}{*}{Problem} & \multicolumn{2}{c|}{Previous result}  & \multirow{2}{*}{Our result}\\
\cline{2-3}
 & $m < \alpha$ & $m \geq \alpha$ & \\
\hline\hline
$S|r_j=0,d_j|E$ & $2(2-\frac{1}{m})^{\alpha}$ \cite{Albers07} & $\min\{2(2-\frac{1}{m})^{\alpha}, B_{\alpha}\}$ \cite{Albers07,Greiner09} & $(2-\frac{1}{m})^{\alpha-1}$\\
$S|r_j,d_j=d|E$ & $2(2-\frac{1}{m})^{\alpha}$ \cite{Albers07} & $\min\{2(2-\frac{1}{m})^{\alpha}, B_{\alpha}\}$ \cite{Albers07,Greiner09} & $(2-\frac{1}{m})^{\alpha-1}$\\
$S|\mbox{clique}|E$ & -- & -- & $(2(2-\frac{1}{m}))^{\alpha-1}$\\
$S|\mbox{agreeable}|E$ & $\alpha^{\alpha}2^{4\alpha}$ \cite{Albers07} & $B_{\alpha}$ \cite{Greiner09} & $(4(2-\frac{1}{m}))^{\alpha-1} < 2^{3\alpha-3}$ \\
\hline
\end{tabular}
\caption{Previous known approximation ratios vs. our approximation ratios.}
\label{tbl:results}
\end{center}
\end{table}

\subsection{Notation and Preliminaries}

%In this paper, we consider that a set $\mathcal{J}$ of $n$ jobs has to be scheduled non-preemptively into a set of $m \geq 1$ speed-scalable processors.
%Each job $J_j \in \mathcal{J}$ is characterized by its processing volume (work) $w_j$, its release date $r_j$ and its deadline $d_j$.
%The speed of a processor is equal to the work executed by this processor per unit of time.
%If any processor operates at a speed $s$, then its power (energy consumption rate) is equal to $s^{\alpha}$.
%In other words, the power of a processor is equal to the energy consumed by this processor per unit of time.

%As in most scheduling problems in the speed-scaling model, for the problems that we study in this paper, it is easy to see that in any optimal schedule,
%any job $J_j \in \mathcal{J}$ runs at a constant speed $s_j$ due to the convexity of the speed-to-power function.

A fundamental property of optimal schedules in the speed-scaling model, which is also true for the problems we study in this paper,
is that any job $J_j \in \mathcal{J}$ runs at a constant speed $s_j$ due to the convexity of the speed-to-power function.
Given a schedule $\mathcal{S}$ and a job $J_j \in \mathcal{J}$,
we denote by $E(\mathcal{S}, J_j)=w_j s_j^{\alpha-1}$ the energy consumed by the execution of $J_j$ in $\mathcal{S}$
and by $E(\mathcal{S})=\sum_{j=1}^n E(\mathcal{S}, J_j)$ the total energy consumed by $\mathcal{S}$.
Given an instance $\mathcal{I}$ and a job $J_j \in \mathcal{J}$,
we denote by $w_j(\mathcal{I})$ the work, by $r_j(\mathcal{I})$ the release date and by $d_j(\mathcal{I})$ the deadline of $J_j$ in $\mathcal{I}$.
We denote by $\mathcal{S}^*$ an optimal non-preemptive schedule for the input instance $\mathcal{I}$.
For each job $J_j \in \mathcal{J}$, we call the interval $[r_j,d_j]$ the \emph{active interval} of $J_j$.

The following proposition has been proved in \cite{Antoniadis12} for $S1|r_j,d_j|E$
but holds also for the corresponding problem on parallel processors.

\begin{proposition} \label{prop:convexity}
\emph{\cite{Antoniadis12}}
Suppose that the schedules $\mathcal{S}$ and $\mathcal{S}'$ process job $J_j$ with speed $s$ and $s'$ respectively.
Assume that $s \leq \gamma s'$ for some $\gamma \geq 1$. Then $E(\mathcal{S},j) \leq \gamma^{\alpha-1} E(\mathcal{S}',j)$.
\end{proposition}

\section{From Preemptive to Non-preemptive Scheduling}

%In this section, we deal with some basic instances, which however are known to be $\mathcal{NP}$-hard,
%and we present approximation algorithms for them.
%The idea of these algorithms is to create the optimal preemptive schedule
%and then to transform it into a non-preemptive schedule, guaranteeing that the energy consumption does not increase more than a factor $\rho$.
%The optimal preemptive schedule can be constructed using one of the algorithms in \cite{Albers11,Angel12,Bingham08}.
%As the energy consumed by an optimal preemptive schedule is a lower bound to the energy consumed by an optimal non-preemptive schedule for the same instance,
%it follows that we have a $\rho$-approximation algorithm.

In this paper, we explore the idea of transforming an optimal preemptive schedule
into a non-preemptive schedule, guaranteeing that the energy consumption of the latter one does not increase more than a factor $\rho$.
Unfortunately, in the following proposition
we show that for general instances the ratio between an optimal non-preemptive schedule to an optimal preemptive one can be very large.

\begin{proposition}\label{prop:neg}
The ratio of the energy consumption of an optimal non-preemptive schedule
to the energy consumption of an optimal preemptive schedule for the same instance of the single-processor speed-scaling problem can be $O(n^{\alpha-1})$.
\end{proposition}
\begin{proof}
Consider the instance consisting of a single processor, $n-1$ unit-work jobs, $J_1,J_2,\ldots,$ $J_{n-1}$, and the job $J_n$ of work $n$.
Each job $J_j$, $1 \leq j \leq n-1$, has release date $r_j=2j-1$ and deadline $d_j=2j$, while $r_n=0$ and $d_n=2n-1$
(see Figure~\ref{fig:pvsnp}).

The optimal preemptive schedule $\mathcal{S}_{pr}$ for this instance assigns to all jobs a speed equal to one.
Each job $J_j$, $1 \leq j \leq n-1$, is executed during its whole active interval,
while $J_n$ is executed during the remaining $n$ unit length intervals.
The total energy consumption of this schedule is $E(\mathcal{S}_{pr})=(n-1) \cdot 1^{\alpha} + n \cdot 1^{\alpha} = 2n-1$.

An optimal non-preemptive schedule $\mathcal{S}_{npr}$ for this instance assigns a speed $\frac{n+2}{3}$ to jobs $J_1$, $J_n$ and $J_2$
and schedules them non-preemptively in this order between time 1 and 4.
Moreover, in $\mathcal{S}_{npr}$ each job $J_j$, $3 \leq j \leq n-1$, is assigned a speed equal to one and it is executed during its whole active interval.
The total energy consumption of this schedule is $E(\mathcal{S}_{npr})= 3 \cdot (\frac{n+2}{3})^{\alpha} + (n-3) \cdot 1^{\alpha}$.

Therefore, we have
$\frac{E(\mathcal{S}_{npr})}{E(\mathcal{S}_{pr})} = \frac{3 \cdot (\frac{n+2}{3})^{\alpha} + (n-3) \cdot 1^{\alpha}}{2n-1} = O(n^{\alpha-1})$.
\end{proof}

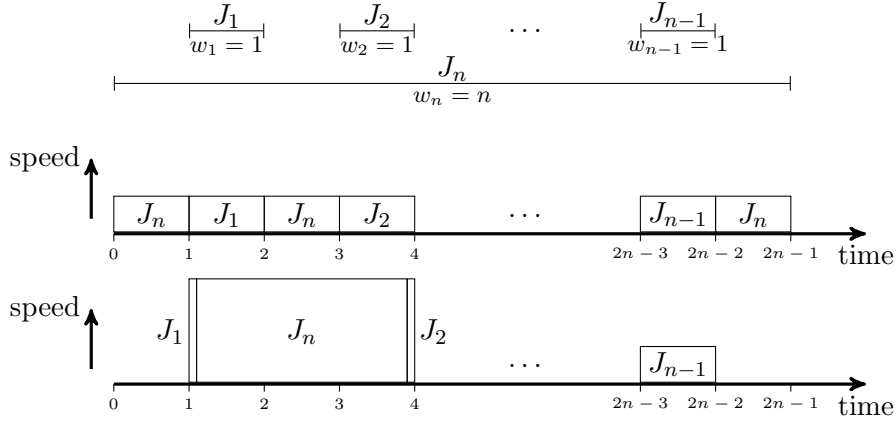
\begin{figure}[htb]
\begin{center}
\begin{tikzpicture}[
    scale=1,
    line/.style={|-|, >=stealth'},
    dline/.style={dashed, thick, >=stealth'},
    axis/.style={very thick, ->, >=stealth'}
    ]
\draw[line] (0.0,2.0) -- (9.0,2.0);
\node at (4.5,2.2) {$J_n$};
\node[font=\footnotesize] at (4.5,1.8) {$w_n=n$};
\draw[line] (1.0,2.7) -- (2.0,2.7);
\node at (1.5,2.9) {$J_1$};
\node[font=\footnotesize] at (1.5,2.5) {$w_1=1$};
\draw[line] (3.0,2.7) -- (4.0,2.7);
\node at (3.5,2.9) {$J_2$};
\node[font=\footnotesize] at (3.5,2.5) {$w_2=1$};
\node at (5.5,2.7) {$\dots$};
\draw[line] (7.0,2.7) -- (8.0,2.7);
\node at (7.5,2.9) {$J_{n-1}$};
\node[font=\footnotesize] at (7.5,2.5) {$w_{n-1}=1$};
\draw[axis] (0,0)  -- (10,0) node(xline)[below] {time};
\draw[axis] (-0.3,0.2) -- (-0.3,1.0) node(yline)[left] {speed};
\draw (0,0) -- (0,-0.1);
\draw (1,0) -- (1,-0.1);
\draw (2,0) -- (2,-0.1);
\draw (3,0) -- (3,-0.1);
\draw (4,0) -- (4,-0.1);
\draw (7,0) -- (7,-0.1);
\draw (8,0) -- (8,-0.1);
\draw (9,0) -- (9,-0.1);
\node[font=\tiny] at (0,-0.27) {0};
\node[font=\tiny] at (1,-0.27) {1};
\node[font=\tiny] at (2,-0.27) {2};
\node[font=\tiny] at (3,-0.27) {3};
\node[font=\tiny] at (4,-0.27) {4};
\node[font=\tiny] at (7,-0.27) {$2n-3$};
\node[font=\tiny] at (8,-0.27) {$2n-2$};
\node[font=\tiny] at (9,-0.27) {$2n-1$};
\draw (0,0.025) rectangle (1,0.5);
\node at (0.5,0.25) {$J_n$};
\draw (2,0.025) rectangle (3,0.5);
\node at (2.5,0.25) {$J_n$};
\draw (8,0.025) rectangle (9,0.5);
\node at (8.5,0.25) {$J_n$};
\draw (1.0,0.025) rectangle (2.0,0.5);
\node at (1.5,0.25) {$J_1$};
\draw (3.0,0.025) rectangle (4.0,0.5);
\node at (3.5,0.25) {$J_2$};
\draw (7.0,0.025) rectangle (8.0,0.5);
\node at (7.5,0.25) {$J_{n-1}$};
\node at (5.5,0.25) {$\dots$};
\draw[axis] (0,-2)  -- (10,-2) node(xline)[below] {time};
\draw[axis] (-0.3,-1.8) -- (-0.3,-1.0) node(yline)[left] {speed};
\draw (0,-2) -- (0,-2.1);
\draw (1,-2) -- (1,-2.1);
\draw (2,-2) -- (2,-2.1);
\draw (3,-2) -- (3,-2.1);
\draw (4,-2) -- (4,-2.1);
\draw (7,-2) -- (7,-2.1);
\draw (8,-2) -- (8,-2.1);
\draw (9,-2) -- (9,-2.1);
\node[font=\tiny] at (0,-2.27) {0};
\node[font=\tiny] at (1,-2.27) {1};
\node[font=\tiny] at (2,-2.27) {2};
\node[font=\tiny] at (3,-2.27) {3};
\node[font=\tiny] at (4,-2.27) {4};
\node[font=\tiny] at (7,-2.27) {$2n-3$};
\node[font=\tiny] at (8,-2.27) {$2n-2$};
\node[font=\tiny] at (9,-2.27) {$2n-1$};
\draw (1.1,-1.975) rectangle (3.9,-0.6);
\node at (2.5,-1.3) {$J_n$};
\draw (1.0,-1.975) rectangle (1.1,-0.6);
\node at (0.75,-1.3) {$J_1$};
\draw (3.9,-1.975) rectangle (4.0,-0.6);
\node at (4.25,-1.3) {$J_2$};
\draw (7.0,-1.975) rectangle (8.0,-1.5);
\node at (7.5,-1.75) {$J_{n-1}$};
\node at (5.5,-1.75) {$\dots$};
\end{tikzpicture}
\caption{An instance of $S1|r_j,d_j|E$ for which the ratio of the energy consumption in an optimal non-preemptive schedule
to the energy consumption in an optimal preemptive schedule is $O(n^{\alpha-1})$.}
\label{fig:pvsnp}
\end{center}
\end{figure}

%The previous result does not allow the use of the cost of an optimal preemptive schedule as a lower bound of the energy consumption of
%an optimal non-preemptive schedule for general instances.
In what follows, we show that for some particular families of instances,
for which the energy minimization multi-processor speed-scaling problem is known to be $\mathcal{NP}$-hard,
it is possible to use the cost of an optimal preemptive schedule as a lower bound of the energy consumption of
an optimal non-preemptive schedule in order to obtain good approximation ratios.

%deal with some basic instances, which however are known to be $\mathcal{NP}$-hard,
%and we present approximation algorithms for them.
%
%The optimal preemptive schedule can be constructed using one of the algorithms in \cite{Albers11,Angel12,Bingham08}.
%As the energy consumed by an optimal preemptive schedule is a lower bound to the energy consumed by an optimal non-preemptive schedule for the same instance,
%it follows that we have a $\rho$-approximation algorithm.
%
%Before begging, we present the following negative result that gives an intuition of the distance between an optimal preemptive
%and an optimal non-preemptive schedule for the same instance of the single-processor speed-scaling problem.

\section{Common Release Dates or Common Deadlines}
\label{section:crd-cd}

%In \cite{Albers07}, it is proved that $S|r_j=0,d_j=d,pmtn,no\text{-}mig|E$ is $\mathcal{NP}$-hard.
%As we mentioned above, any optimal solution for this problem can be transformed into an optimal solution for $S|r_j=0,d_j=d|E$,
%and hence the latter problem is also $\mathcal{NP}$-hard.
In this section we deal with $S|r_j=0,d_j|E$ and $S|r_j,d_j=d|E$, which are $\mathcal{NP}$-hard as generalizations of $S|r_j=0,d_j=d|E$.
In fact, we will present an approximation algorithm for $S|r_j=0,d_j|E$,
which achieves a better ratio than the known approximation algorithms presented in \cite{Albers07} and \cite{Greiner09} for any values of $\alpha$ and $m$.
We also describe how to adapt this algorithm to an algorithm for $S|r_j,d_j=d|E$ of the same approximation ratio.

\algo{CRD} takes as input an instance $\mathcal{I}$ of $S|r_j=0,d_j|E$ and creates first an optimal preemptive schedule for $\mathcal{I}$,
using one of the algorithms in \cite{Albers11,Angel12,Bingham08}.
The total execution time $e_j$ of each job $J_j \in \mathcal{J}$ in this preemptive schedule is used to define an appropriate processing time $p_j$ for $J_j$.
Then, the algorithm schedules non-preemptively the jobs using these processing times according to the Earliest Deadline First policy,
i.e., at every time that a machine becomes idle, the non-scheduled job with the minimum deadline is scheduled on it.
The choice of the values of the $p_j$'s has been made in such a way that the algorithm completes all the jobs before their deadlines.

\begin{algorithm}[h!]
\algo{CRD($\mathcal{I}$)}
\begin{algorithmic}[1]
\STATE Create an optimal preemptive schedule $\mathcal{S}_{pr}$ for $\mathcal{I}$;
\STATE Let $e_j$ be the total execution time of the job $J_j \in \mathcal{J}$, in $\mathcal{S}_{pr}$;
\STATE Schedule the jobs with the Earliest Deadline First (EDF) policy,
       using the appropriate speed such that the processing time of the job $J_j \in \mathcal{J}$, is equal to $p_j=e_j/(2-\frac{1}{m})$,
       obtaining the non-preemptive schedule $\mathcal{S}_{npr}$;
\RETURN $\mathcal{S}_{npr}$;
\end{algorithmic}
\end{algorithm}

\begin{theorem} \label{thm:crd}
\algo{CRD} is a $(2-\frac{1}{m})^{\alpha-1}$-approximation algorithm for $S|r_j=0,d_j|E$.
\end{theorem}
\begin{proof}
We first prove by induction that for the completion time $C_j$ of the job $J_j \in \mathcal{J}$, it holds that $C_j \leq d_j$.

For each job $J_j$, $1 \leq j \leq m$, it holds that $C_j = p_j = e_j/(2-\frac{1}{m}) \leq d_j$,
since $J_j$ is selected by the algorithm for $S|r_j,d_j,pmtn|E$ such that $\mathcal{S}_{pr}$ to be feasible, and hence $e_j \leq d_j$.
Assume that our hypothesis is true for all jobs in $\mathcal{J}_{k-1}=\{J_1,J_2,\ldots,J_{k-1}\}$.

For the job $J_k$, it holds that $C_k \leq \frac{\sum_{j=1}^{k-1} p_j}{m} + p_k = (\frac{\sum_{j=1}^{k-1} e_j}{m} + e_k) / (2-\frac{1}{m})$.
As $\mathcal{S}_{pr}$ is a feasible schedule and $d_k \geq d_j$ for each job $J_j \in \mathcal{J}_{k-1}$,
it holds that $e_k \leq d_k$ and $\sum_{j=1}^k e_j \leq m \cdot d_k$.
Hence, $C_k \leq (2-\frac{1}{m}) d_k / (2-\frac{1}{m}) = d_k$ and $\mathcal{S}_{npr}$ is a feasible schedule.

Moreover, when dividing the execution time of all jobs by $(2-\frac{1}{m})$,
at the same time the speed of each job is multiplied by the same factor.
Thus, using Proposition~\ref{prop:convexity} we have that
\begin{equation*}
E(\mathcal{S}_{npr}) \leq (2-\frac{1}{m})^{\alpha-1} E(\mathcal{S}_{pr}) \leq (2-\frac{1}{m})^{\alpha-1} E(\mathcal{S}^*)
\end{equation*}
since the energy consumed by the optimal preemptive schedule $\mathcal{S}_{pr}$
is a lower bound to the energy consumed by an optimal non-preemptive schedule $\mathcal{S}^*$ for the input instance $\mathcal{I}$.
\end{proof}

Next, we describe how to transform \algo{CRD} for $S|r_j=0,d_j|E$ into \algo{CD} for $S|r_j,d_j=d|E$.
In order to do this, we modify the Line~3 of \algo{CRD} such that to use the Latest Release Date First (LRDF) policy,
scheduling the jobs in a backward way starting from the deadline $d$.
Using a similar analysis, the following theorem holds.

\begin{theorem} \label{thm:cd}
\algo{CD} is a $(2-\frac{1}{m})^{\alpha-1}$-approximation algorithm for $S|r_j,d_j=d|E$.
\end{theorem}

\section{Clique Instances}
\label{section:clique}

In this section, we present a constant factor approximation algorithm for $S|\mbox{clique}|E$.
Recall that this problem is $\mathcal{NP}$-hard as a generalization of $S|r_j=0,d_j=d|E$.
Moreover, in \cite{Antoniadis12} it is proved that $S1|\mbox{pure-laminar}|E$ is $\mathcal{NP}$-hard.
Hence, even $S1|\mbox{clique}|E$ is $\mathcal{NP}$-hard.

\algo{Cl} takes as input a clique instance $\mathcal{I}$ and creates first an optimal preemptive schedule,
using again one of the algorithms in \cite{Albers11,Angel12,Bingham08}.
Taking into account the execution times of jobs before and after the time $T=\min\{d_j, J_j \in \mathcal{J}\}$ in the preemptive schedule,
the algorithm splits the set of jobs into two subsets: the left $\mathbb{L}$ and the right $\mathbb{R}$.
The set $\mathbb{L}$ contains the jobs which have bigger execution time before $T$,
while the set $\mathbb{R}$ consists of the jobs of bigger execution time after $T$.
For the jobs in $\mathbb{L}$, we modify their deadlines to $T$ and we use \algo{CD}.
For the jobs in $\mathbb{R}$, we modify their release dates to $T$ and we use \algo{CRD}.
Then, \algo{Cl} returns the concatenation of the two schedules created by \algo{CD} and \algo{CRD}.

\begin{algorithm}[h!]
\algo{Cl($\mathcal{I}$)}
\begin{algorithmic}[1]
\STATE Create an optimal preemptive schedule $\mathcal{S}_{pr}$ for $\mathcal{I}$;
\STATE Let $e_{\ell j}$ and $e_{rj}$ be the total execution time of each job $J_j \in \mathcal{J}$, before and after, respectively,
       the time $T=\min\{d_j, J_j \in \mathcal{J}\}$ in $\mathcal{S}_{pr}$;
\STATE Let $\mathbb{L}=\{J_j: e_{\ell j} \geq e_{rj}\}$ and $\mathbb{R}=\{J_j: e_{\ell j} < e_{rj}\}$;
\STATE Consider the instance $\mathcal{I}_{\mathbb{L}}$ consisting of the jobs in $\mathbb{L}$,
       and each $J_j \in \mathbb{L}$ is characterized by a work $w_j(\mathcal{I}_{\mathbb{L}})=w_j(\mathcal{I})$,
       a release date $r_j(\mathcal{I}_{\mathbb{L}})=r_j(\mathcal{I})$,
       and a deadline $d_j(\mathcal{I}_{\mathbb{L}})=T$;
\STATE Run \algo{CD($\mathcal{I}_{\mathbb{L}})$} obtaining the non-preemptive schedule $\mathcal{S}_{npr}^{\mathbb{L}}$;
\STATE Consider the instance $\mathcal{I}_{\mathbb{R}}$ consisting of the jobs in $\mathbb{R}$,
       and each $J_j \in \mathbb{R}$ is characterized by a work $w_j(\mathcal{I}_{\mathbb{R}})=w_j(\mathcal{I})$,
       a release date $r_j(\mathcal{I}_{\mathbb{R}})=T$,
       and a deadline $d_j(\mathcal{I}_{\mathbb{R}})=d_j(\mathcal{I})$;
\STATE Run \algo{CRD($\mathcal{I}_{\mathbb{R}})$} obtaining the non-preemptive schedule $\mathcal{S}_{npr}^{\mathbb{R}}$;
\STATE Create the schedule $\mathcal{S}_{npr}$ by concatenating the schedules $\mathcal{S}_{npr}^{\mathbb{L}}$ and $\mathcal{S}_{npr}^{\mathbb{R}}$;
\RETURN $\mathcal{S}_{npr}$;
\end{algorithmic}
\end{algorithm}

\begin{theorem}\label{thm:clique}
\algo{Cl} is a $(2(2-\frac{1}{m}))^{\alpha-1}$-approximation algorithm for $S|\mbox{clique}|E$.
\end{theorem}
\begin{proof}
Let $\mathcal{S}_{pr}^{\mathbb{L}}$ and $\mathcal{S}_{pr}^{\mathbb{R}}$ be the preemptive schedules
created in Line~1 of \algo{CD} and \algo{CRD}, respectively.
By Theorems~\ref{thm:cd} and \ref{thm:crd}, respectively, it holds that
\begin{equation}\label{eqn:clique1}
E(\mathcal{S}_{npr}) = E(\mathcal{S}_{npr}^{\mathbb{L}}) + E(\mathcal{S}_{npr}^{\mathbb{R}})
                      \leq (2-\frac{1}{m})^{\alpha-1} (E(\mathcal{S}_{pr}^{\mathbb{L}})+E(\mathcal{S}_{pr}^{\mathbb{R}}))
\end{equation}

Consider, now, the preemptive schedule $\mathcal{S}$ that occurs from $\mathcal{S}_{pr}$
if we schedule the whole work of each job $J_j \in \mathbb{L}$ (resp. $J_j \in \mathbb{R}$) before (resp. after) the time $T$
at the intervals in which $J_j$ is already executed in $\mathcal{S}_{pr}$.
By construction for each job $J_j \in \mathbb{L}$ (resp. $\mathbb{R}$) it holds that $e_{\ell j} \geq e_{rj}$ (resp. $e_{\ell j} < e_{rj}$).
Hence, in order $\mathcal{S}$ to be feasible, we have just to double the speed of each job $J_j \in \mathcal{J}$,
and thus by Proposition~\ref{prop:convexity} we have $E(\mathcal{S},J_j)=2^{\alpha-1}E(\mathcal{S}_{pr},J_j)$.
In total we get that $E(\mathcal{S})=2^{\alpha-1}E(\mathcal{S}_{pr})$.

Let $\mathcal{S}^{\mathbb{L}}$ and $\mathcal{S}^{\mathbb{R}}$ be the parts of $\mathcal{S}$ before and after $T$, respectively,
that is $E(\mathcal{S})=E(\mathcal{S}^{\mathbb{L}})+E(\mathcal{S}^{\mathbb{R}})$.
The schedules $\mathcal{S}^{\mathbb{L}}$ and $\mathcal{S}_{pr}^{\mathbb{L}}$ concern the same instance,
and since $\mathcal{S}_{pr}^{\mathbb{L}}$ is the optimal one, it holds that $E(\mathcal{S}_{pr}^{\mathbb{L}}) \leq E(\mathcal{S}^{\mathbb{L}})$.
Similarly, it holds that $E(\mathcal{S}_{pr}^{\mathbb{R}}) \leq E(\mathcal{S}^{\mathbb{R}})$.
Therefore, we get
\begin{equation}\label{eqn:clique2}
E(\mathcal{S}_{pr}^{\mathbb{L}})+E(\mathcal{S}_{pr}^{\mathbb{R}}) \leq
E(\mathcal{S}^{\mathbb{L}})+E(\mathcal{S}^{\mathbb{R}}) =
E(\mathcal{S}) = 2^{\alpha-1} E(\mathcal{S}_{pr})
\end{equation}
By combining Equations~(\ref{eqn:clique1}) and~(\ref{eqn:clique2})
and taking into account that $E(\mathcal{S}_{pr})$ is a lower bound to the energy consumed by an optimal non-preemptive schedule $\mathcal{S}^*$ for $\mathcal{I}$,
i.e., $E(\mathcal{S}_{pr}) \leq E(\mathcal{S}^*)$, the theorem follows.
\end{proof}

\begin{corollary}
\algo{Cl} is a $(2(2-\frac{1}{m}))^{\alpha-1}$-approximation algorithm for $S|\mbox{pure-laminar}|E$.
\end{corollary}

\section{Agreeable Instances}
\label{section:agreeable}

In this section we deal with agreeable instances and
for $\alpha > m$ we improve the known $\alpha^{\alpha}2^{4\alpha}$-approximation algorithm
presented in \cite{Albers07} for $S|\mbox{agreeable}|E$, by proposing a $(4(2-\frac{1}{m}))^{\alpha-1}$-approximation algorithm.
Recall that, $S|\mbox{agreeable}|E$ is known to be $\mathcal{NP}$-hard as a generalization of $S|r_j=0,d_j=d|E$.
On the other hand, $S1|\mbox{agreeable}|E$ can be solved in polynomial time \cite{Antoniadis12,Yao95}.
In order to handle an agreeable instance $\mathcal{I}$, we will first describe how to partition it into clique subinstances,
by splitting the set of jobs $\mathcal{J}$ into subsets.

To begin with this partition, we define $T_1$ to be the smallest deadline of any job in $\mathcal{J}$, i.e., $T_1=\min\{d_j : J_j \in \mathcal{J}\}$.
Let $\mathcal{J}_1 \subseteq \mathcal{J}$ be the set of jobs which are released before $T_1$, i.e., $\mathcal{J}_1=\{J_j \in \mathcal{J} : r_j \leq T_1\}$.
Next, we set $T_2=\min\{d_j : J_j \in \mathcal{J} \setminus \mathcal{J}_1\}$ and $\mathcal{J}_2=\{J_j \in \mathcal{J} : T_1 < r_j \leq T_2\}$,
and we continue this process until all jobs are assigned into a subset of jobs (see Figure~\ref{fig:agreeable}).
Let $k$ be the number of subsets of jobs that have been created.
%Moreover, let $T_0=\min\{r_j: J_j \in \mathcal{J}\}$ and $T_{k+1}=\max\{d_j: J_j \in \mathcal{J}\}$.

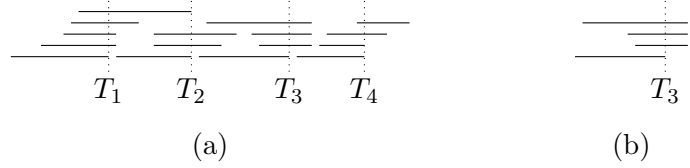
\begin{figure}[htb]
\begin{center}
\begin{tikzpicture}
\draw (0.0,0.00) -- (1.3,0.00);
\draw (0.4,0.15) -- (1.4,0.15);
\draw (0.7,0.30) -- (1.4,0.30);
\draw (0.8,0.45) -- (1.7,0.45);
\draw (0.9,0.60) -- (2.4,0.60);
\draw (1.4,0.00) -- (2.4,0.00);
\draw (1.9,0.15) -- (2.8,0.15);
\draw (1.9,0.30) -- (3.0,0.30);
\draw (2.5,0.00) -- (3.7,0.00);
\draw (3.3,0.15) -- (4.0,0.15);
\draw (3.2,0.30) -- (4.0,0.30);
\draw (2.6,0.45) -- (4.0,0.45);
\draw (3.8,0.00) -- (4.7,0.00);
\draw (4.1,0.15) -- (4.7,0.15);
\draw (4.2,0.30) -- (5.0,0.30);
\draw (4.6,0.45) -- (5.3,0.45);
\draw[dotted] (1.3,-0.2) -- (1.3,0.8);
\node at (1.3,-0.45) {$T_1$};
\draw[dotted] (2.4,-0.2) -- (2.4,0.8);
\node at (2.4,-0.45) {$T_2$};
\draw[dotted] (3.7,-0.2) -- (3.7,0.8);
\node at (3.7,-0.45) {$T_3$};
\draw[dotted] (4.7,-0.2) -- (4.7,0.8);
\node at (4.7,-0.45) {$T_4$};
\node at (2.65,-1.2) {(a)};
\draw (7.5,0.00) -- (8.7,0.00);
\draw (8.3,0.15) -- (9.0,0.15);
\draw (8.2,0.30) -- (9.0,0.30);
\draw (7.6,0.45) -- (9.0,0.45);
\draw[dotted] (8.7,-0.2) -- (8.7,0.8);
\node at (8.7,-0.45) {$T_3$};
\node at (8.25,-1.2) {(b)};
\end{tikzpicture}
\caption{(a) The partition of jobs. (b) The subset of jobs $\mathcal{J}_3$.}
\label{fig:agreeable}
\end{center}
\end{figure}

In the above partition, the subsets of jobs are pairwise disjoint.
However, the two clique subinstances induced by two subsets of jobs may have time interference.
In other words, there may exist two jobs $J_j \in \mathcal{J}_{\ell}$ and $J_{j'} \in \mathcal{J}_{\ell'}$,
where $0 \leq \ell < \ell' \leq k$, such that $d_j \geq r_{j'}$.
Nevertheless, the key observation here is that in the above partition into clique subinstances
only two consecutive subsets of jobs may have time interference.
In fact a stronger property holds: for each $\ell$, $1 \leq \ell \leq k-1$, and each job $J_j \in \mathcal{J}_{\ell}$,
the active interval of $J_j$ does not contain $T_{\ell+1}$.
To see this, given a partition into the subsets of jobs $\mathcal{J}_1,\mathcal{J}_2,\ldots,\mathcal{J}_k$ and the corresponding times $T_1,T_2,\ldots,T_k$,
assume for contradiction that there is a job $J_j \in \mathcal{J}_{\ell}$ whose active interval contains $T_{\ell+1}$.
Hence, it holds that $r_j \leq T_{\ell}$ and $d_j > T_{\ell+1}$.
Consider, now, the job $J_{j'} \in \mathcal{J}_{\ell+1}$ whose deadline defines $T_{\ell+1}$.
For $J_{j'}$ it holds that $r_{j'} > T_{\ell}$ and $d_{j'}=T_{\ell+1}$.
Thus, we have that $r_j < r_{j'}$ and $d_j > d_{j'}$, which is a contradiction that our instance is agreeable.

Next, we propose \algo{Agr} which initially defines the times $T_1,T_2,\ldots,T_k$ and
the corresponding subsets of jobs $\mathcal{J}_1,\mathcal{J}_2,\ldots,\mathcal{J}_k$.
Then, in order to separate the intervals of two consecutive subsets of jobs,
our algorithm decreases the active intervals of all jobs by half.
This will ensure the feasibility of the schedule for the whole instance $\mathcal{I}$.
For each clique subinstance, \algo{Cl} is called and it gives a non-preemptive schedule.
Our algorithm returns the concatenation of the schedules for the clique subinstances.

\begin{algorithm}[h!]
\algo{Agr}
\begin{algorithmic}[1]
\STATE Find the partition into the subsets of jobs $\mathcal{J}_1,\mathcal{J}_2,\ldots,\mathcal{J}_k$ and the corresponding times $T_1,T_2,\ldots,T_k$;
\FOR {$\ell=1$ to $k$}
\STATE Consider the clique instance $\mathcal{I}_{\ell}$ consisting of the jobs in $\mathcal{J}_{\ell}$,
       and each $J_j \in \mathcal{J}_{\ell}$ is characterized by a work $w_j(\mathcal{I}_{\ell})=w_j(\mathcal{I})$,
       a release date $r_j(\mathcal{I}_{\ell})=r_j(\mathcal{I})+\frac{T_{\ell}-r_j(\mathcal{I})}{2}$,
       and a deadline $d_j(\mathcal{I}_{\ell})=d_j(\mathcal{I})-\frac{d_j(\mathcal{I})-T_{\ell}}{2}$;
\STATE Run \algo{Cl($\mathcal{I}_{\ell}$)} obtaining the non-preemptive schedule $\mathcal{S}_{npr}^{(\ell)}$;
\ENDFOR
\STATE Create the schedule $\mathcal{S}_{npr}$
       by concatenating the schedules $\mathcal{S}_{npr}^{(1)},\mathcal{S}_{npr}^{(2)},\ldots,\mathcal{S}_{npr}^{(k)}$;
\RETURN $\mathcal{S}_{npr}$;
\end{algorithmic}
\end{algorithm}

The following proposition deals with two clique instances consisting of the same set of jobs with different active intervals
(as in Line~3 of \algo{Agr}).
Note that the proposition holds both for preemptive and non-preemptive schedules.

\begin{proposition} \label{prop:clique}
Consider a clique instance $\mathcal{I}_1$ consisting of a set of jobs $\mathcal{J}_1$ and let $T=\min\{d_j:~J_j \in \mathcal{J}_1\}$.
Let $\mathcal{I}_2$ be an instance consisting of the set of jobs $\mathcal{J}_2=\mathcal{J}_1$
and for each $J_j \in \mathcal{J}_2$ we have that $w_j(\mathcal{I}_2)=w_j(\mathcal{I}_1)$,
$r_j(\mathcal{I}_2)=r_j(\mathcal{I}_1)+\frac{T-r_j(\mathcal{I}_1)}{2}$ and $d_j(\mathcal{I}_2)=d_j(\mathcal{I}_1)-\frac{d_j(\mathcal{I}_1)-T}{2}$.
For two optimal schedules $\mathcal{S}_1$ and $\mathcal{S}_2$ for $\mathcal{I}_1$ and $\mathcal{I}_2$, respectively,
it holds that $E(\mathcal{S}_2) \leq 2^{\alpha-1} E(\mathcal{S}_1)$.
\end{proposition}
\begin{proof}
Note first that the active interval of each job $J_j \in \mathcal{J}_1$ from $d_j(\mathcal{I}_1)-r_j(\mathcal{I}_1)$ in $\mathcal{I}_1$
becomes $d_j(\mathcal{I}_2)-r_j(\mathcal{I}_2)=\frac{1}{2}(d_j(\mathcal{I}_1)-r_j(\mathcal{I}_1))$ in $\mathcal{I}_2$.

Given an optimal schedule $\mathcal{S}_1$ for $\mathcal{I}_1$, we can obtain a feasible schedule $\mathcal{S}_2'$ for $\mathcal{I}_2$
if we decrease the execution time of each job $J_j \in \mathcal{J}_1$ by a factor of $\frac{1}{2}$
and then compact the schedule keeping fixed the time $T$.
In order to do this, we have to increase the speed of $J_j$ by a factor of $2$.
Hence, by Proposition~\ref{prop:convexity} it holds that $E(\mathcal{S}_2') = 2^{\alpha-1} E(\mathcal{S}_1)$.
As $\mathcal{S}_2$ is the optimal schedule for $\mathcal{I}_2$, the proposition follows.
\end{proof}

\begin{theorem}\label{thm:agreeable}
\algo{Agr}
is a $(4(2-\frac{1}{m}))^{\alpha-1}$-approximation algorithm for $S|\mbox{agreeable}|E$.
\end{theorem}
\begin{proof}
We first deal with the approximation ratio of the algorithm.
Let $\mathcal{S}_{pr}^{(1)},\mathcal{S}_{pr}^{(2)},\ldots,\mathcal{S}_{pr}^{(k)}$ be the optimal preemptive schedules created in Line~1
of \algo{Cl} with input the instances $\mathcal{I}_1,\mathcal{I}_2,\ldots,\mathcal{I}_k$, respectively.
By Theorem~\ref{thm:clique}, for each $\ell$, $1 \leq \ell \leq k$,
it holds that $E(\mathcal{S}_{npr}^{(\ell)}) \leq (2(2-\frac{1}{m}))^{\alpha-1} E(\mathcal{S}_{pr}^{(\ell)})$, and hence
\begin{equation} \label{eqn:small1}
E(\mathcal{S}_{npr}) = \sum_{\ell=1}^k E(\mathcal{S}_{npr}^{(\ell)}) \leq (2(2-\frac{1}{m}))^{\alpha-1} \sum_{\ell=1}^k E(\mathcal{S}_{pr}^{(\ell)})
\end{equation}

For each $\ell$, $1 \leq \ell \leq k$, let $\mathcal{S}^{(\ell)}$ be the optimal preemptive schedule for the jobs in $\mathcal{J}_{\ell}$
subject to their original release dates and deadlines in $\mathcal{I}$.
By Proposition~\ref{prop:clique} we have that $E(\mathcal{S}_{pr}^{(\ell)}) \leq 2^{\alpha-1} E(\mathcal{S}^{(\ell)})$, and hence
\begin{equation} \label{eqn:small2}
\sum_{\ell=1}^k E(\mathcal{S}_{pr}^{(\ell)}) \leq 2^{\alpha-1} \sum_{\ell=1}^k E(\mathcal{S}^{(\ell)})
\end{equation}

Let $\mathcal{S}_{pr}$ be an optimal preemptive schedule for the whole instance $\mathcal{I}$.
Note that, $\sum_{\ell=1}^k E(\mathcal{S}^{(\ell)})$ is a lower bound to $E(\mathcal{S}_{pr})$,
since all jobs in $\bigcup_{\ell=1}^k \mathcal{S}^{(\ell)}$ have their original release dates and deadlines,
and each of $\mathcal{S}^{(\ell)}$ is an optimal preemptive schedule.
Hence, by combining the Equations~(\ref{eqn:small1}) and~(\ref{eqn:small2}), the approximation ratio achieved by our algorithm follows.

In order to show the feasibility of the schedule $\mathcal{S}_{npr}$ created by \algo{Agr}, it suffices to show that
for each $\ell$, $1 \leq \ell \leq k-1$, and for any two jobs $J_j \in \mathcal{J}_{\ell}$ and $J_{j'} \in \mathcal{J}_{\ell+1}$
it holds that $d_j(\mathcal{I}_{\ell}) < r_{j'}(\mathcal{I}_{\ell+1})$.
More intuitively, we want to show that the active intervals of any two jobs belonging to two consecutive (w.r.t. time) instances does not intersect.

By construction, it holds that $d_j(\mathcal{I}) \leq T_{\ell+1}$ and $r_{j'}(\mathcal{I}) > T_{\ell}$.
Moreover, we have that $d_j(\mathcal{I}_{\ell}) = d_j(\mathcal{I}) - \frac{d_j(\mathcal{I})-T_{\ell}}{2} \leq \frac{T_{\ell+1}+T_{\ell}}{2}$
and $r_{j'}(\mathcal{I}_{\ell+1}) = r_{j'}(\mathcal{I}) + \frac{T_{\ell+1}-r_{j'}(\mathcal{I})}{2} > \frac{T_{\ell+1}+T_{\ell}}{2}$.
Therefore, $d_j(\mathcal{I}_{\ell}) < r_{j'}(\mathcal{I}_{\ell+1})$, and the theorem follows.
\end{proof}

%\section{Concluding Remarks}

%We have improved the approximation ratio for the non-preemptive speed-scaling problem on a single processor
%(e.g., for $\alpha=2$ we improve $64$ to $23.31$, and for $\alpha=3$ we improve $2048$ to $184.67$).
%In order to do this, we have simplified the ideas of \cite{Antoniadis12} and combined them with two important ingredients:
%the appropriate partition of jobs into ``big'' and ``small'', and the use of the cost of an optimal preemptive schedule
%as a lower bound for the energy consumption of ``small'' jobs.
%We have also significantly improve the ratios given in \cite{Albers07} for instances with common release dates or deadlines by a factor of $2(2-\frac{1}{m})$.
%Two interesting questions for future work are (i) whether there is a constant factor approximation algorithm for general instances for the multiprocessor case
%and (ii) whether it is possible to improve our results for the single-processor case.

%\bibliographystyle{plain}
%\bibliography{non-preemptive}

\begin{thebibliography}{10}

\bibitem{Albers10}
S.~Albers.
\newblock Energy-efficient algorithms.
\newblock {\em Communications of ACM}, 53:86--96, 2010.

\bibitem{Albers11}
S.~Albers, A.~Antoniadis, and G.~Greiner.
\newblock On multi-processor speed scaling with migration: extended abstract.
\newblock In {\em 23rd ACM Symposium on Parallelism in Algorithms and Architectures (SPAA 2011)}, pages 279--288. ACM, 2011.

\bibitem{Albers07}
S.~Albers, F.~M\"{u}ller, and S.~Schmelzer.
\newblock Speed scaling on parallel processors.
\newblock In {\em 19th ACM Symposium on Parallelism in Algorithms and Architectures (SPAA 2007)}, pages 289--298. ACM, 2007.

\bibitem{Angel12}
E.~Angel, E.~Bampis, F.~Kacem, and D.~Letsios.
\newblock Speed scaling on parallel processors with migration.
\newblock In {\em 18th International European Conference on Parallel and Distributed Computing (Euro-Par 2012)}, volume 7484 of {\em LNCS}, pages 128--140.
  Springer, 2012.

\bibitem{Antoniadis12}
A.~Antoniadis and C.-C. Huang.
\newblock Non-preemptive speed scaling.
\newblock In {\em 13th Scandinavian Symposium and Workshops on Algorithm Theory (SWAT 2012)}, volume 7357 of {\em LNCS}, pages 249--260.
  Springer, 2012.

\bibitem{Azar05}
Y.~Azar and A.~Epstein.
\newblock Convex programming for scheduling unrelated parallel machines.
\newblock In {\em 37th annual ACM Symposium on Theory of Computing (STOC 2005)}, pages 331--337, 2005.

\bibitem{Bampis12}
E.~Bampis, D.~Letsios, I.~Milis, and G.~Zois.
\newblock Speed scaling for maximum lateness.
\newblock In {\em 18th Annual International Computing and Combinatorics Conference (COCOON'12)}, volume 7434 of {\em LNCS}, pages 25--36.
  Springer, 2012.

\bibitem{Bingham08}
B.~D. Bingham and M.~R. Greenstreet.
\newblock Energy optimal scheduling on multiprocessors with migration.
\newblock In {\em International Symposium on Parallel and Distributed Processing with Applications (ISPA 2008)}, pages 153--161. IEEE, 2008.

\bibitem{Bunde06}
D.~P. Bunde.
\newblock Power-aware scheduling for makespan and flow.
\newblock In {\em 18th ACM Symposium on Parallelism in Algorithms and Architectures (SPAA 2006)}, pages 190--196. ACM, 2006.

\bibitem{Chen04}
J.-J. Chen, H.-R. Hsu, K.-H, Chuang, C.-L. Yang, A.-C. Pang, and T.-W. Kuo.
\newblock Multiprocessor energy-efficient scheduling with task migration
  considerations.
\newblock In {\em 16th Euromicro Conference of Real-Time Systems (ECTRS 2004)}, pages 101--108, 2004.

\bibitem{Greiner09}
G.~Greiner, T.~Nonner, and A.~Souza.
\newblock The bell is ringing in speed-scaled multiprocessor scheduling.
\newblock In {\em 21st ACM Symposium on Parallelism in Algorithms and Architectures (SPAA 2009)}, pages 11--18. ACM, 2009.

\bibitem{Hochbaum87}
D.S. Hochbaum and D.B. Shmoys.
\newblock Using dual approximation algorithms for scheduling problems:
  {T}heoretical and practical results.
\newblock {\em Journal of the ACM}, 34:144--162, 1987.

\bibitem{Pruhs08}
K.~Pruhs, R.~van Stee, and P.~Uthaisombut.
\newblock Speed scaling of tasks with precedence constraints.
\newblock {\em Theory of Computing Systems}, 43:67--80, 2008.

\bibitem{Yao95}
F.~Yao, A.~Demers, and S.~Shenker.
\newblock A scheduling model for reduced {CPU} energy.
\newblock In {\em 36th Annual Symposium on Foundations of Computer Science (FOCS 1995)}, pages 374--382, 1995.

\end{thebibliography}

\end{document}